\documentclass[11pt]{amsart}
\usepackage{mathrsfs}
\usepackage{amssymb} %Extra symbols 
\usepackage{url}
\usepackage{hyperref}
\usepackage{color}
\definecolor{refcolor}{RGB}{0,0,190}
\hypersetup{
    colorlinks,
    citecolor=refcolor,
    filecolor=refcolor,
    linkcolor=refcolor,
    urlcolor=refcolor
}

\theoremstyle{definition}
\newtheorem{theorem}{Theorem}[section]
\newtheorem{definition}[theorem]{Definition}
\newtheorem{proposition}[theorem]{Proposition}

\newtheorem{corollary}[theorem]{Corollary}
\newtheorem{remark}[theorem]{Remark}

\newtheorem{openproblem}[theorem]{Open Problem}

\def\({\left(}
\def\){\right)}

\newcommand{\R}{\mathbb{R}}

\newcommand{\de}{\textnormal{d}}

\newcommand{\grad}{\textnormal{grad }}

\newcommand{\tn}{\textnormal}
\newcommand{\ds}{\displaystyle}

\newcommand{\ie}{\textit{i.e.} }

\newcommand{\eg}{\textit{e.g.} }

\newcommand{\citep}[2]{\cite{#1}, p. #2}

\newcommand{\Ric}{\textnormal{Ric}}

\newcommand{\mf}[1]{\mathfrak{#1}}
\newcommand{\mc}[1]{\mathcal{#1}}

\newcommand{\sref}[1]{\S\ref{#1}}

\newcommand{\fivectlift}[1]{\mf L(#1)}
\newcommand{\ric}{\tn{Ric}}

\newcommand{\dsfrac}[2]{\ds{\frac{#1}{#2}}}

\newcommand{\metric}[1]{\langle#1\rangle}
\newcommand{\kosz}{\mc K}
\newcommand{\vectmodule}{\mf X}
\newcommand{\fivect}[1]{\vectmodule(#1)}
\newcommand{\cocontr}{{{}_\bullet}}

\newcommand{\annihg}{{g}_{\bullet}}

\def\hyph{-\penalty0\hskip0pt\relax}
\newcommand{\semiriem}{semi{\hyph}Riemannian}

\newcommand{\semireg}{semi{\hyph}regular}
\newcommand{\ssemireg}{Semi{\hyph}regular}
\newcommand{\quasireg}{quasi{\hyph}regular}
\newcommand{\qquasireg}{Quasi{\hyph}regular}
\newcommand{\nondeg}{non{\hyph}degenerate}

\newcommand{\flrw}{Friedmann-Lema\^itre-Robertson-Walker}
\newcommand{\FLRW}{FLRW}
\newcommand{\schw}{Schwarzschild}
\newcommand{\rn}{Reissner-Nordstr\"om}
\newcommand{\kn}{Kerr-Newman}

\begin{document} 
 
%--------------------------------------------------------
% Title
\title{Einstein equation at singularities}
\author{Cristi \ Stoica}
\date{\today. Horia Hulubei National Institute for Physics and Nuclear Engineering, Bucharest, Romania. E-mail: holotronix@gmail.com}

\begin{abstract}
Einstein's equation is rewritten in an equivalent form, which remains valid at the singularities in some major cases. These cases include the Schwarzschild singularity, the Friedmann-Lema\^itre-Robertson-Walker Big Bang singularity, isotropic singularities, and a class of warped product singularities. This equation is constructed in terms of the Ricci part of the Riemann curvature (as the Kulkarni-Nomizu product between Einstein's equation and the metric tensor).
\bigskip
\noindent 
\end{abstract}

%--------------------------------------------------------
% Title and contents

\maketitle

\setcounter{tocdepth}{1}
\tableofcontents

%--------------------------------------------------------
\section*{Introduction}

The singularities in General Relativity can be avoided only if the stress-energy tensor in the right hand side of Einstein's equation satisfies some particular conditions. One way to avoid them was proposed by the authors of \cite{corda2010removingBHsingularities}, who have shown that the singularities can be removed by constructing the stress-energy tensor with non-linear electrodynamics. On the other hand, Einstein's
equation leads to singularities in general conditions \cite{Pen65,Haw66i,Haw66ii,Haw67iii,HP70,HE95}, and there the time evolution breaks down. Is this a problem of the theory itself or of the way it is formulated?

This paper proposes a version of Einstein's equation which is equivalent to the standard version at the points of spacetime where the metric is non-singular. But unlike Einstein's equation, in many cases it can be extended at and beyond the singular points.

Let $(M,g)$ be a Riemannian or a {\semiriem} manifold of dimension $n$.
It is useful to recall the definition of the \textit{Kulkarni-Nomizu product} of two symmetric bilinear forms $h$ and $k$,
\begin{equation}
\label{eq_kulkarni_nomizu}
	(h\circ k)_{abcd} := h_{ac}k_{bd} - h_{ad}k_{bc} + h_{bd}k_{ac} - h_{bc}k_{ad}.
\end{equation}
The Riemann curvature tensor can be decomposed algebraically as
\begin{equation}
\label{eq_ricci_decomposition}
	R_{abcd} = S_{abcd} + E_{abcd} + C_{abcd}.
\end{equation}
where 
\begin{equation}
\label{eq_ricci_part_S}
	S_{abcd} = \dsfrac{1}{2n(n-1)}R(g\circ g)_{abcd}
\end{equation}
is the scalar part of the Riemann curvature and 
\begin{equation}
\label{eq_ricci_part_E}
	E_{abcd} = \dsfrac{1}{n-2}(S \circ g)_{abcd}
\end{equation}
is the \textit{semi-traceless part} of the Riemann curvature. Here
\begin{equation}
\label{eq_ricci_traceless}
S_{ab} := R_{ab} - \dsfrac{1}{n}Rg_{ab}
\end{equation}
is the traceless part of the Ricci curvature.

The \textit{Weyl curvature tensor} is defined as the \textit{traceless part} of the Riemann curvature 
\begin{equation}
\label{eq_weyl_curvature}
	C_{abcd} = R_{abcd} - S_{abcd} - E_{abcd}.
\end{equation}

The Einstein equation is
\begin{equation}
\label{eq_einstein}
	G_{ab} + \Lambda g_{ab} = \kappa T_{ab},
\end{equation}
where $T_{ab}$ is the stress-energy tensor of the matter, the constant $\kappa$ is defined as $\kappa:=\dsfrac{8\pi \mc G}{c^4}$, where $\mc G$ and $c$ are the gravitational constant and the speed of light, and $\Lambda$ is the \textit{cosmological constant}. 
The term
\begin{equation}
\label{eq_einstein_tensor}
	G_{ab}:=R_{ab}-\frac 1 2 R g_{ab}
\end{equation}
is the Einstein tensor, constructed from the \textit{Ricci curvature} $R_{ab} := g^{st}R_{asbt}$ and the \textit{scalar curvature} $R := g^{st}R_{st}$.

As it is understood, the Einstein equation establishes the connection between curvature and stress-energy. The curvature contributes to the equation in the form of the Ricci tensor $R_{ab}$ and the scalar curvature. In the proposed equation, the curvature contributes in the form of the semi-traceless and scalar parts of the Riemann tensor, $E_{abcd}$ \eqref{eq_ricci_part_E} and $S_{abcd}$ \eqref{eq_ricci_part_S}, which are tensors of the same order and have the same symmetries as $R_{abcd}$. 

The Ricci tensor $R_{ab}$ is obtained by contracting the tensor $E_{abcd}+S_{abcd}$, and has the same information (if the metric is {\nondeg}). One can move from the fourth-order tensors $E_{abcd}+S_{abcd}$ to $R_{ab}$ by contraction, and one can move back to them by taking the Kulkarni-Nomizu product \eqref{eq_kulkarni_nomizu}, but they are equivalent. Yet, if the metric $g_{ab}$ is degenerate, then $g^{ab}$ and the contraction $R_{ab}=g^{st}(E_{asbt}+S_{asbt})$ become divergent, even if $g_{ab}$, $E_{abcd}$, and $S_{abcd}$ are smooth. This suggests the possibility that $E_{abcd}$ and $S_{abcd}$ are more fundamental that the Ricci and scalar curvatures.

This suggestion is in agreement with the following observation. In the case of \textit{electrovac} solutions, where $F_{ab}$ is the electromagnetic tensor,
\begin{equation}
\label{eq_stress_energy_maxwell}
	T_{ab}=\frac{1}{4\pi}\(\frac 1 4 g_{ab} F_{st}F^{st} - F_{as} F_b{}^s\)=-\frac{1}{8\pi}\(F_{ac}F_b{}^c + {}^\ast F_{ac} {}^\ast F_b{}^c\),
\end{equation}
where ${}^\ast$ is the Hodge duality operation. It can be obtained by contracting the semi-traceless part of the Riemann tensor
\begin{equation}
\label{eq_stress_energy_maxwell_expanded}
	E_{abcd}=-\frac{\kappa}{8\pi}\(F_{ab}F_{cd} + {}^\ast F_{ab} {}^\ast F_{cd}\).
\end{equation}

Therefore it is natural to at least consider an equation in terms of these fourth-order tensors, rather than the Ricci and scalar curvatures.

The main advantage of this method is that there are singularities in which the new formulation of the Einstein equation is not singular (although the original Einstein equation exhibits singularities, obtained when contracting with the singular tensor $g^{ab}$). The expanded Einstein equation is written in terms of the smooth geometric objects $E_{abcd}$ and $S_{abcd}$. Because of this the solutions can be extended at singularities where the original Einstein equation diverges. This doesn't mean that the singularities are removed; for example the Kretschmann scalar $R_{abcd}R^{abcd}$ is still divergent at some of these singularities. But this is not a problem, since the Kretschmann scalar is not part of the evolution equation. It is normally used as an indicator that there is a singularity, for example to prove that the {\schw} singularity at $r=0$ cannot be removed by coordinate changes, as the event horizon singularity can. While a singularity of the Kretschmann scalar indicates the presence of a singularity of the curvature, it doesn't have implications on whether the singularity can be resolved or not. In the proposed equation we use $R_{abcd}$ which is smooth at the studied singularities, and we don't use $R^{abcd}$ which is singular and causes the singularity of the Kretschmann scalar.

A second reason to consider the expanded version of the Einstein equation and the {\quasireg} singularities at which it is smooth is that at these singularities the Weyl curvature tensor vanishes. The implications of this feature will be explored in \cite{Sto12c}.

It will be seen that there are some important examples of singularities which turn out to be {\quasireg}. While singularities still exist, our approach provides a description in terms of smooth geometric objects which remain finite at singularities. By this we hope to improve our understanding of singularities and to distinguish those to which our resolution applies.

The \textit{expanded Einstein equations} and the {\quasireg} spacetimes on which they hold are introduced in section \sref{s_einstein_exp_qreg}. They are obtained  by taking the Kulkarni-Nomizu product between Einstein's equation and the metric tensor. In a {\quasireg} spacetime the metric tensor becomes degenerate at singularities in a way which cancels them and makes the equations smooth.

The situations when the new version of Einstein's equation extends at singularities include isotropic singularities (section \sref{s_qreg_examples_isotropic}) and a class of warped product singularities (section \sref{s_qreg_examples_warped}). It also contains the {\schw} singularity (section \sref{s_qreg_examples_schw}) and the {\FLRW} Big Bang singularity (section \sref{s_qreg_examples_flrw}).

%--------------------------------------------------------
\section{Expanded Einstein equation and {\quasireg} spacetimes}
\label{s_einstein_exp_qreg}

%--------------------------------------------------------
\subsection{The expanded Einstein equation}
\label{s_einstein_exp}

An equation which is equivalent to Einstein's equation whenever the metric tensor $g_{ab}$ is {\nondeg}, but is valid also in a class of situations when $g_{ab}$ becomes degenerate and Einstein's tensor is not defined will be discussed in this section. Later it will be shown that the proposed version of Einstein's equation remains smooth in various important situations such as the FLRW Big-Bang singularity, isotropic singularities, and at the singularity of the {\schw} black hole.

We introduce the \textit{expanded Einstein equation}
\begin{equation}
\label{eq_einstein_expanded}
	(G\circ g)_{abcd} + \Lambda (g\circ g)_{abcd} = \kappa (T\circ g)_{abcd}.
\end{equation}

If the metric is {\nondeg} then the Einstein equation and its expanded version are equivalent. This can be seen by contracting the expanded Einstein equation, for instance in the indices $b$ and $d$. From \eqref{eq_kulkarni_nomizu} the contraction in $b$ and $d$ of a Kulkarni-Nomizu product $(h\circ g)_{abcd}$ is
\begin{equation}
\hat h_{ac}:=(h\circ g)_{asct}g^{st} = h_{ac}g^s_s - h_{at}\delta^t_c + h^s_sg_{ac} - h_{sc}\delta^s_a = 2h_{ac} + h^s_sg_{ac}.
\end{equation}
From $\hat h_{ac}$ the original tensor $h_{ac}$ can be obtained again by
\begin{equation}
\label{eq_expanded_to_standard}
h_{ac}=\frac 1 2 \hat h_{ac} - \frac 1{12}\hat h{}^s_s g_{ac}.
\end{equation}
By this procedure the terms $G_{ab}$, $T_{ab}$, and $\Lambda g_{ab}$ can be recovered from the equation \eqref{eq_einstein_expanded}, thus obtaining the Einstein equation \eqref{eq_einstein}
. Hence, the Einstein equation and its expanded version are equivalent for a {\nondeg} metric.

If the metric becomes degenerate its inverse becomes singular, and in general the Riemann, Ricci, and scalar curvatures, and consequently the Einstein tensor $G_{ab}$, diverge. For certain cases the metric term from the Kulkarni-Nomizu product $G\circ g$ tends to $0$ fast enough to cancel the divergence of the Einstein tensor. The {\quasireg} singularities satisfy the condition that the divergence of $G$ is compensated by the degeneracy of the metric, so that $G\circ g$ is smooth.

This cancellation allows us to weaken the condition that the metric tensor is {\nondeg}, to some cases when it can be degenerate. It will be seen that these cases include some important singularities.

%--------------------------------------------------------
\subsection{A more explicit form of the expanded Einstein equation}
\label{s_einstein_exp_explicit}

To give a more explicit form of the expanded Einstein equation, the \textit{Ricci decomposition} of the Riemann curvature tensor is used (see \eg \cite{ST69,BESS87,GHLF04}).

By using the equations \eqref{eq_einstein_tensor}
 and \eqref{eq_ricci_traceless} in dimension $n=4$, the Einstein tensor in terms of the traceless part of the Ricci tensor and the scalar curvature can be written:
\begin{equation}
G_{ab} = S_{ab} - \dsfrac{1}{4}R g_{ab}.
\end{equation}

This equation can be used to calculate the \textit{expanded Einstein tensor}:
\begin{equation}
\label{eq_einstein_tensor_expanded}
\begin{array}{lrl}
G_{abcd} &:=& (G\circ g)_{abcd} \\
&=& (S \circ g)_{abcd} - \dsfrac{1}{4}R (g\circ g)_{abcd}\\
&=& 2 E_{abcd} - 6 S_{abcd}.
\end{array}
\end{equation}
The expanded Einstein equation now takes the form
\begin{equation}
\label{eq_einstein_expanded_explicit}
	2 E_{abcd} - 6 S_{abcd} + \Lambda (g\circ g)_{abcd} = \kappa (T\circ g)_{abcd}.
\end{equation}

%--------------------------------------------------------
\subsection{Quasi-regular spacetimes}
\label{s_qreg_spacetimes}

We are interested in singular spacetimes on which the expanded Einstein equation \eqref{eq_einstein_expanded} can be written and is smooth. From \eqref{eq_einstein_expanded_explicit} it can be seen that this requires the smoothness of the tensors $E_{abcd}$ and $S_{abcd}$.

In addition we are interested to have the nice properties of the {\semireg} spacetimes.
As showed in \cite{Sto11a}, the {\semireg} manifolds are a class of singular {\semiriem} manifolds which are nice for several reasons, one of them being that the Riemann tensor $R_{abcd}$ is smooth.

First, a contraction between covariant indices is needed. This is in general prohibited by the fact that when the metric tensor $g_{ab}$ becomes degenerate it doesn't admit a reciprocal $g^{ab}$.
Although the metric $g_{ab}$ can't induce an invariant inner product on the cotangent space $T_p^*M$, it induces one on its subspace $\flat(T_pM)$, where $\flat:T_pM\to T_p^*M$ is the vector space morphism defined by $X^\flat(Y):=\metric{X,Y}$, for any $X,Y\in T_pM$. Equivalently, $\flat(T_pM)$ is the space of $1$-forms $\omega$ on $T_pM$ so that $\omega|_{\ker\flat}=0$. The morphism $\flat$ is isomorphism if and only if $g$ is {\nondeg}; in this case its inverse is denoted by $\sharp$.
The inner product on $\flat(T_pM)$ is then defined by $\annihg(X^\flat,Y^\flat):=\metric{X,Y}$ and it is invariant.
This allows us to define a contraction between covariant slots of a tensor $T$, which vanishes when vectors from $\ker\flat$ are plugged in those slots.
This will turn out to be enough for our needs.
We denote the contractions between covariant indices of a tensor $T$ by $T(\omega_1,\ldots,\omega_r,v_1,\ldots,\cocontr,\ldots,\cocontr,\ldots,v_s)$.

A degenerate metric also prohibits in general the construction of a Levi-Civita connection. For vector fields we use instead of $\nabla_XY$, the \textit{Koszul form}, defined as:
\begin{equation*}
	\kosz:\fivect M^3\to\R,
\end{equation*}
\begin{equation}
\label{eq_Koszul_form}
	\kosz(X,Y,Z) :=\ds{\frac 1 2} \{ X \metric{Y,Z} + Y \metric{Z,X} - Z \metric{X,Y} 
	- \metric{X,[Y,Z]} + \metric{Y, [Z,X]} + \metric{Z, [X,Y]}\}
\end{equation}
which defines the Levi-Civita connection by $\nabla_XY=\kosz(X,Y,\_)^\sharp$ for a {\nondeg} metric, but not when the metric becomes degenerate.
We define now {\semireg} manifolds, on which we can define covariant derivatives for a large class of differential forms and tensors. We can also define a generalization of the Riemann curvature $R_{abcd}$, which turns out to be smooth and non-singular.

\begin{definition}
\label{def_semi_regular}
A singular {\semiriem} manifold satisfying the condition that $\kosz(X,Y,\_)\in\flat(T_pM)$, and that the contraction $\kosz(X,Y,\cocontr)\kosz(Z,T,\cocontr)$ is smooth for any local vector fields $X,Y,Z,T$, is named \textit{{\semireg} manifold}, and its metric is called \textit{{\semireg} metric}.
A $4$-dimensional {\semireg} manifold with metric having the signature at each point $(r,s,t)$, $s\leq 3$, $t\leq 1$, but which is {\nondeg} on a dense subset, is called \textit{{\semireg} spacetime} \cite{Sto11a}. 
\end{definition}

In \cite{Sto11a} we defined the Riemann curvature $R_{abcd}$ for {\semireg} metrics, even for {\nondeg} metrics, in a way which avoids the undefined $\nabla_XY$, but relies on the defined and smooth $\kosz(X,Y,Z)$, by
\begin{equation}
\label{eq_riemann_curvature_tensor_coord}
	R_{abcd}= \partial_a \Gamma_{bcd} - \partial_b \Gamma_{acd} + \Gamma_{ac\cocontr}\Gamma_{bd\cocontr} - \Gamma_{bc\cocontr}\Gamma_{ad\cocontr},
\end{equation}
where $\Gamma_{abc}=\kosz(\partial_a,\partial_b,\partial_c)$ are the Christoffel's symbols of the first kind. From Definition \ref{def_semi_regular}, $R_{abcd}$ is smooth. 
More details on the {\semireg} manifolds can be found in  \cite{Sto11a,Sto11b,Sto12e}.

In a {\semireg} spacetime, since $R_{abcd}$ is smooth, the densitized Einstein tensor $G_{ab}\det g$ is smooth \cite{Sto11a}, and a densitized version of the Einstein equation can be written, which is equivalent to the usual version when the metric is {\nondeg}:
\begin{equation}
\label{eq_einstein_idx:densitized}
	G_{ab}\sqrt{-g}^W + \Lambda g_{ab}\sqrt{-g}^W = \kappa T_{ab}\sqrt{-g}^W,
\end{equation}
where it is enough to take the weight $W\leq 2$.
Although the {\semireg} approach is more general, here is explored the {\quasireg} one, which is more strict. Consequently, these results are stronger.

\begin{definition}
\label{def_quasi_regular}
We say that a {\semireg} manifold $(M,g_{ab})$ is \textit{{\quasireg}}, and that $g_{ab}$ is a \textit{{\quasireg} metric}, if:
\begin{enumerate}
	\item 
$g_{ab}$ is {\nondeg} on a subset dense in $M$
	\item 
	the tensors $S_{abcd}$ and $E_{abcd}$ defined at the points where the metric is {\nondeg} extend smoothly to the entire manifold $M$.
\end{enumerate}
If the {\quasireg} manifold $M$ is a {\semireg} spacetime, we call it \textit{{\quasireg} spacetime}. Singularities of {\quasireg} manifolds are called {\quasireg}.
\end{definition}

It can be seen that on an {\quasireg} spacetime the expanded Einstein tensor can be extended at the points where the metric is degenerate, and the extension is smooth. This is in fact the motivation of Definition \ref{def_quasi_regular}.

%--------------------------------------------------------
\section{Examples of {\quasireg} spacetimes}
\label{s_qreg_examples}

The {\quasireg} spacetimes are more general than the regular ones (those with {\nondeg} metric), containing them as a particular case. The question is, are they general enough to cover the singularities which plagued General Relativity? In the following it will be seen that at least for some relevant cases the answer is positive. It will be seen that the class of {\quasireg} singularities contain isotropic singularities \sref{s_qreg_examples_isotropic}, singularities obtained as warped products \sref{s_qreg_examples_warped} (including the {\flrw} spacetime \sref{s_qreg_examples_flrw}), and even the {\schw} singularity \sref{s_qreg_examples_schw}. The existence of these examples which are extensively researched justifies the study of the more general {\quasireg} singularities and of the extended Einstein equations.

%--------------------------------------------------------
\subsection{Isotropic singularities}
\label{s_qreg_examples_isotropic}

\textit{Isotropic singularities} occur in conformal rescalings of {\nondeg} metrics, when the scaling function cancels. They were extensively studied by Tod \cite{Tod87,Tod90,Tod91,Tod92,Tod02,Tod03}, Claudel \& Newman \cite{CN98}, Anguige \& Tod \cite{AT99i,AT99ii}, in connection with cosmological models. The following theorem shows that the isotropic singularities are {\quasireg}.

\begin{theorem}[Isotropic singularities]
\label{thm_quasireg_example_conformal}
Let $(M,g_{ab})$ be a regular spacetime (we assume therefore that the metric $g_{ab}$ is {\nondeg}). Then, if $\Omega:M\to\R$ is a smooth function which is non-zero on a dense subset of $M$, the spacetime $(M,\widetilde g_{ab} :=\Omega^2 g_{ab})$ is {\quasireg}.
\end{theorem}
\begin{proof}
From \cite{Sto11a} is known that $(M,\widetilde g_{ab})$ is {\semireg}.

The Ricci and the scalar curvatures take the following forms (\cite{HE95}, p. 42.):
\begin{equation}
\label{eq_conformal_ricci_curv_ud}
\widetilde R^a{}_b = \Omega^{-2}R^a{}_b + 2\Omega^{-1}(\Omega^{-1})_{;bs}g^{as}-\dsfrac 1 2\Omega^{-4}(\Omega^2)_{;st}g^{st}\delta^a{}_b
\end{equation}
\begin{equation}
\label{eq_conformal_scalar_curv}
\widetilde R=\Omega^{-2}R-6\Omega^{-3}\Omega_{;st}g^{st}
\end{equation}
where the covariant derivatives correspond to the metric $g$.
From equation \eqref{eq_conformal_ricci_curv_ud} follows that
\begin{equation}
\widetilde R_{ab}=\Omega^2 g_{as} \widetilde R^s{}_b=R_{ab} + 2\Omega(\Omega^{-1})_{;ab}-\dsfrac 1 2\Omega^{-2}(\Omega^2)_{;st}g^{st}g_{ab},
\end{equation}
which tends to infinity when $\Omega\to 0$. But we are interested to prove the smoothness of the Kulkarni-Nomizu product $\widetilde\Ric\circ \widetilde g$. We notice that the term $\widetilde g$ contributes with a factor $\Omega^2$, and it is enough to prove the smoothness of
\begin{equation}
\Omega^2\widetilde R_{ab}=\Omega^2 R_{ab} + 2\Omega^3(\Omega^{-1})_{;ab}-\dsfrac 1 2(\Omega^2)_{;st}g^{st}g_{ab},
\end{equation}
which follows from
\begin{equation}
\begin{array}{lll}
\Omega^3(\Omega^{-1})_{;ab} &=& \Omega^3\((\Omega^{-1})_{;a}\)_{;b} = \Omega^3\(-\Omega^{-2}\Omega_{;a}\)_{;b} \\
&=& \Omega^3\(2\Omega^{-3}\Omega_{;b}\Omega_{;a} - \Omega^{-2}\Omega_{;ab}\) \\
&=& 2\Omega_{;a}\Omega_{;b} - \Omega\Omega_{;ab} \\
\end{array}
\end{equation}
Hence, the tensor $\widetilde\Ric\circ \widetilde g$ is smooth. The fact that $\widetilde R \widetilde g\circ \widetilde g$ is smooth follows from the observation that $\widetilde g\circ \widetilde g$ contributes with $\Omega^4$, and the least power in which $\Omega$ appears in the expression \eqref{eq_conformal_scalar_curv} of $\widetilde R$ is $-3$.

From the above follows that $\widetilde E_{abcd}$ and $\widetilde S_{abcd}$ are smooth. Hence the spacetime $(M,\widetilde g_{ab})$ is {\quasireg}.
\end{proof}

%--------------------------------------------------------
\subsection{{\qquasireg} warped products}
\label{s_qreg_examples_warped}

Another example useful in cosmology is the following, which is a generalization of the warped products. Warped products are extensively researched, since they allow the construction of {\semiriem} spacetimes, having applications to GR. But when the warping function becomes $0$, singularities occur (see \eg \citep{ONe83}{ 204}). Fortunately, in the cases of interest for General Relativity, these singularities are {\quasireg}. We will allow the warped function $f$ to become $0$ (generalizing the standard definition \cite{ONe83}, where it is not allowed to vanish because it leads to degenerate metrics), and prove that what the resulting singularities are {\quasireg}.

\begin{definition}
\label{def_wp}
Let $(B,\de s_B^2)$ and $(F,\de s_F^2)$ be two {\semiriem} manifolds, and $f: B\to\R$ a smooth function on $B$. The \textit{degenerate warped product} of $B$ and $F$ with \textit{warping function} $f$ is the manifold $B\times_f F:=\(B\times F,\de s_{B\times F}^2\)$, with the metric
\begin{equation}
\de s_{B\times F}^2 = \de s_B^2 + f^2\de s_F^2
\end{equation}
\end{definition}

\begin{theorem}[{\qquasireg} warped product]
\label{thm_quasireg_example_wp}
A degenerate warped product $B\times_f F$ with $\dim B=1$ is {\quasireg}.
\end{theorem}
\begin{proof}
From \cite{Sto11b}, $B\times_f F$ is {\semireg}.

Let's denote by $g_B$, $g_F$ and $g$ the metrics on $B$, $F$ and $B\times_f F$.
It is known (\cite{ONe83}, p. 211) that for horizontal vector fields $X,Y\in\fivectlift{B \times F,B}$ and vertical vector fields $V,W\in\fivectlift{B \times F,F}$,
\begin{enumerate}
	\item $\ric(X,Y) = \ric_B(X,Y) + \dsfrac{\dim F}{f}H^f(X,Y)$
	\item $\ric(X,V) = 0$
	\item $\ric(V,W) = \ric_F(V,W) + \(f\Delta f + (\dim F-1)g_B(\grad f,\grad f)\)g_F(V,W)$
\end{enumerate}
where $\Delta f$ is the Laplacian, $H^f$ the Hessian, and $\grad f$ the gradient.
It follows that $\ric(X,V)$ and $\ric(V,W)$ are smooth, but $\ric(X,Y)$ in general is not, because of the term containing $f^{-1}$. But since $\dim B=1$, the only terms in the Kulkarni-Nomizu product $\Ric\circ g$ containing $\Ric(X,Y)$ are of the form 
\begin{equation*}
\Ric(X,Y)g(V,W)=f^2\Ric(X,Y)g_F(V,W).
\end{equation*}
Hence, $\Ric\circ g$ is smooth.

From the expression of the scalar curvature
\begin{equation}
\label{eq_scalar_curv_wp}
	R = R_B + \frac {R_F}{f^2} + 2\dim F\dsfrac{\Delta f}{f} + \dim F(\dim F - 1)\dsfrac{g_B(\grad f,\grad f)}{f^2}
\end{equation}
can be concluded that $S_{abcd}$ is smooth too, because $g\circ g$ contains at least one factor of $f^2$. Hence, $B\times_f F$ is {\quasireg}.
\end{proof}

The following example important in cosmology is a direct application of this result.

\begin{proposition}[{\ssemireg} manifold which is not {\quasireg}]
Let $B=\R^k$, $k>1$, be an Euclidean space, with the canonical metric $g_B$, and $f:B\to\R$ a linear function $f\neq 0$. Let $F=\R^l$, $l>1$, with the canonical metric $g_F$. 
Then the degenerate warped product $B\times_f F$ is {\semireg}, but it isn't {\quasireg}.
\end{proposition}
\begin{proof}
Because $f$ is linear but not constant, $\grad f\neq 0$ is constant, and $\Delta f=0$. % and $H^f=0$. The Ricci tensor of the warped product is smooth, because all terms containing $f$ at denominator vanish (see the formulae $1-3$ from the proof of Theorem \ref{thm_quasireg_example_wp}). 
The scalar curvature \eqref{eq_scalar_curv_wp} becomes $R=l(l - 1)\dsfrac{g_B(\grad f,\grad f)}{f^2}$, which is singular at $0$. Because $k>1$, $g_B\circ g_B$ doesn't vanish, hence it doesn't cancel the denominator $f^2$ of $R$ in the term $R g_B\circ g_B$. Also, the term $R g_B\circ g_B$ is not canceled by other terms composing $S_{abcd}$, because they are all smooth, containing at least one $g_F$. Hence, $S_{abcd}$ is singular, and the degenerate warped product $B\times_f F$ isn't {\quasireg}. On the other hand, according to \cite{Sto11b}, because $B$ and $F$ are {\nondeg}, $B\times_f F$ is {\semireg}.
\end{proof}

%--------------------------------------------------------
\subsection{The {\flrw} spacetime}
\label{s_qreg_examples_flrw}

The {\flrw} ({\FLRW}) spacetime is defined as the warped product $I\times_a \Sigma$, where
\begin{enumerate}
	\item 
	$I\subseteq \R$ is an interval representing the time, which is viewed as a {\semiriem} space with the negative definite metric $-c^2\de t^2$.
	\item 
	$(\Sigma,\de\Sigma^2)$ is a three-dimensional Riemannian space, usually one of the homogeneous spaces $S^3$, $\R^3$, and $H^3$ (to model the homogeneity and isotropy conditions at large scale). Then the metric on $\Sigma$ is, in spherical coordinates $(r,\theta,\phi)$,
\begin{equation}
\label{eq_flrw_sigma_metric}
\de\Sigma^2 = \dsfrac{\de r^2}{1-k r^2} + r^2\(\de\theta^2 + \sin^2\theta\de\phi^2\),
\end{equation}
where $k=1,0,-1$, for the $3$-sphere $S^3$, the Euclidean space $\R^3$, or hyperbolic space $H^3$ respectively.
	\item
	$a: I\to \R$ is a function of time.
\end{enumerate}

The {\FLRW} metric is
\begin{equation}
\label{eq_flrw_metric}
\de s^2 = -c^2\de t^2 + a^2(t)\de\Sigma^2.
\end{equation}

At any moment of time $t\in I$ the space is $\Sigma_t=(\Sigma,a^2(t)g_\Sigma)$.

For a {\FLRW} universe filled with a fluid with mass density $\rho(t)$ and pressure density $p(t)$, the stress-energy tensor is defined as
\begin{equation}
\label{eq_friedmann_stress_energy}
T^{ab} = \(\rho + \dsfrac{p}{c^2}\)u^a u^b + p g^{ab},
\end{equation}
where $g(u,u)=-c^2$.

From Einstein's equation with the stress-energy tensor \eqref{eq_friedmann_stress_energy} follow the \textit{Friedmann equation}
\begin{equation}
\label{eq_friedmann_density}
\rho = \kappa^{-1}\(3\dsfrac{\dot{a}^2 + kc^2}{c^2 a^2} - \Lambda \),
\end{equation}
which gives the mass density $\rho(t)$ in terms of $a(t)$, and the \textit{acceleration equation}
\begin{equation}
\label{eq_acceleration}
\dsfrac{p}{c^2} = \dsfrac{2}{\kappa c^2}\(\dsfrac{\Lambda}{3}-\dsfrac{1}{c^2} \dsfrac{\ddot{a}}{a}\) - \dsfrac \rho 3,
\end{equation}
giving the pressure density $p(t)$.

A question that may arise is what happens with the densities $\rho$ and $p$. Equations \eqref{eq_friedmann_density} and \eqref{eq_acceleration} show that $\rho$ and $p$ may diverge in most cases for $a\to 0$. As explained in \cite{Sto11h}, $\rho$ and $p$ are calculated considering orthonormal frames. If the frame is not necessarily orthonormal (because there is no orthonormal frame at the point where the metric is degenerate), then the volume element is not necessarily equal to $1$, and it has to be included in the equations. The scalars $\rho$ and $p$ are replaced by the differential $4$-forms which have the components $\rho\sqrt{-g}$ and $p\sqrt{-g}$. It can be seen by calculation that these forms are smooth.
If the metric on the manifold $\Sigma$ is denoted by $g_{\Sigma}$, then the Friedmann equation \eqref{eq_friedmann_density} becomes
\begin{equation}
\label{eq_friedmann_density_tilde}
\rho\sqrt{-g} = \dsfrac{3}{\kappa}a\(\dot a^2 + k\) \sqrt{g_{\Sigma}},
\end{equation}
and the acceleration equation \eqref{eq_acceleration} becomes
\begin{equation}
\label{eq_acceleration_tilde}
\rho\sqrt{-g} + 3p\sqrt{-g} = -\dsfrac{6}{\kappa}a^2\ddot{a} \sqrt{g_{\Sigma}},
\end{equation}
hence $\rho\sqrt{-g}$ and $p\sqrt{-g}$ are smooth.

As $a\to 0$, the metric becomes degenerate, $\rho$ and $p$ diverge, and therefore the stress-energy tensor \eqref{eq_friedmann_stress_energy} diverges too. Because of this, the Ricci tensor also diverges. But, from Theorem \ref{thm_quasireg_example_wp}, $R_{abcd}$, $E_{abcd}$, and $S_{abcd}$ are smooth. What can be said about the expanded stress-energy tensor $(T \circ g)_{abcd}$? The following corollary shows that the metric is {\quasireg}, hence the expanded stress-energy tensor is smooth.

\begin{corollary}
\label{thm_flrw}
The {\FLRW} spacetime with smooth $a: I\to \R$ is {\quasireg}.
\end{corollary}
\begin{proof}
Since the {\FLRW} spacetime is a warped product between a $1$-dimensional and a $3$-dimensional manifold with warping function $a$, this is a direct consequence of Theorem \ref{thm_quasireg_example_wp}. 
\end{proof}

\begin{remark}
Corollary \ref{thm_flrw} applies not only to a {\FLRW} universe filled with a fluid, but to more general ones. For this particular case a direct proof was given in \cite{Sto12a}, showing explicitly how the expected infinities of the physical fields cancel out.
\end{remark}

While the expanded Einstein equation for the {\FLRW} spacetime with smooth $a$ is written in terms of smooth objects like $E_{abcd}$, $S_{abcd}$, and $T_{abcd}:=(T \circ g)_{abcd}$, a question arises, as to why use these objects, instead of $R_{ab}$, $S$, and $T_{ab}$? It is true that the expanded objects remain smooth, while the standard ones don't, but is there other, more fundamental reason? It can be said that $E_{abcd}$ and $S_{abcd}$ are more fundamental, since $R_{ab}$ and $R$ are obtained from them by contractions. But for $T_{abcd}$, unfortunately, at this time we don't know an interpretation. The stress-energy tensor $T_{ab}$ can be obtained from a Lagrangian, but we don't know yet a way to obtain directly $T_{abcd}$ from a Lagrangian. One hint that, at least for some fields, $T_{abcd}$ seems more fundamental is that, for electrovac solutions, it is given by $T_{abcd}=-\frac{1}{8\pi}\(F_{ab}F_{cd} + {}^\ast F_{ab} {}^\ast F_{cd}\)$ \eqref{eq_stress_energy_maxwell_expanded}, while $T_{ab}$ by contracting it \eqref{eq_stress_energy_maxwell}. Similar form has the stress-energy tensor for Yang-Mills fields.

Another question that may appear is what is obtained, given that the solution can be extended beyond the moment when $a(t)=0$? Say that $a(0)=0$. The extended solution will describe two universes, both originating from the same Big-Bang at the same moment $t=0$, one of them expanding toward the direction in which $t$ increases and the other one toward the direction in which $t$ decreases. The parameter $t$ is just a coordinate, and the physical laws are symmetric with respect to time reversal in General Relativity (if one wants to consider quantum fields, the combined symmetry $CPT$ should be considered instead of $T$ alone).

%--------------------------------------------------------
\subsection{{\schw} black hole}
\label{s_qreg_examples_schw}

The {\schw} solution describing a black hole of mass $m$ is given in the {\schw} coordinates by the metric tensor:
\begin{equation}
\label{eq_schw_schw}
\de s^2 = -\(1-\dsfrac{2m}{r}\)\de t^2 + \(1-\dsfrac{2m}{r}\)^{-1}\de r^2 + r^2\de\sigma^2,
\end{equation}
where
\begin{equation}
\label{eq_sphere}
\de\sigma^2 = \de\theta^2 + \sin^2\theta \de \phi^2
\end{equation}
is the metric of the unit sphere $S^2$. The units were chosen so that $c=1$ and $G=1$ (see \eg \citep{HE95}{149}).

Apparently the metric is singular at $r=2m$, on the event horizon. As it is known from the work of Eddington \cite{eddington1924comparison} and Finkelstein \cite{finkelstein1958past} appropriate coordinate changes make the metric {\nondeg} on the event horizon, showing that the singularity is apparent, being due to the coordinates. The coordinate change is singular, but it can be said that the proper coordinates around the event horizon are those of Eddington and Finkelstein, and the {\schw} coordinates are the singular coordinates.

Can we apply a similar method for the singularity at $r=0$? It can be checked that the Kretschmann scalar $R_{abcd}R^{abcd}$ is singular at $r=0$, and since scalars are invariant at any coordinate changes (including the singular ones), it is usually correctly concluded that the singularity at $r=0$ cannot be removed. Although it cannot be removed, it can be improved by finding coordinates making the metric analytic at $r=0$. As shown in \cite{Sto11e} the singularity $r=0$ in the {\schw} metric \eqref{eq_schw_schw} has two origins -- it is a combination of degenerate metric and singular coordinates. Firstly, the {\schw} coordinates are singular at $r=0$, but they can be desingularized by applying the coordinate transformations from equation \eqref{eq_coordinate_semireg} which necessarily have the Jacobian equal to zero at $r=0$. It is not possible to desingularize a coordinate system, by using transformations that have non-vanishing Jacobian at the singularity, because such transformations preserve the regularity of the metric. Secondly, after the transformation the singularity is not completely removed, because the metric remains degenerate. However, the metric remains {\semireg}, as shown in \cite{Sto11e}. Here will be shown that it is also {\quasireg}.

In \cite{Sto11e} we showed that the {\schw} solution can be made analytic at the singularity by a coordinate transformation of the form
\begin{equation}
\label{eq_coordinate_change}
\left\{
\begin{array}{ll}
r &= \tau^S \\
t &= \xi\tau^T \\
\end{array}
\right.
\end{equation}
As it turns out,
\begin{equation}
\label{eq_coordinate_semireg}
\left\{
\begin{array}{ll}
r &= \tau^2 \\
t &= \xi\tau^4 \\
\end{array}
\right.
\end{equation}
is the only choice which makes analytic at the singularity not only the metric, but also the Riemann curvature $R_{abcd}$. In the new coordinates the metric has the form
\begin{equation}
\label{eq_schw_semireg}
\de s^2 = -\dsfrac{4\tau^4}{2m-\tau^2}\de \tau^2 + (2m-\tau^2)\tau^4\(4\xi\de\tau + \tau\de\xi\)^2 + \tau^4\de\sigma^2.
\end{equation}

\begin{corollary}
\label{thm_schw_quasireg}
The {\schw} spacetime is {\quasireg} (in any atlas compatible with the coordinates \eqref{eq_coordinate_semireg}).
\end{corollary}
\begin{proof}
We know from \cite{Sto11e} that the {\schw} spacetime is {\semireg}. Since it is also Ricci flat, \ie $R_{ab}=0$, it follows that $S_{ab}=1$ and $R=0$, hence $S_{abcd}= \dsfrac{1}{24}R(g\circ g)_{abcd}=0$, and $E_{abcd}\dsfrac{1}{2}(S \circ g)_{abcd}=0$. Therefore, $S_{abcd}$ and $E_{abcd}$ are smooth. Consequently, the only non-vanishing part of the curvature in the Ricci decomposition \eqref{eq_ricci_decomposition} is the Weyl tensor $C_{abcd}$, which in this case is equal to $R_{abcd}$, so it is smooth too.
\end{proof}

\begin{remark}
It has been seen that even if the {\schw} metric $g_{ab}$ is singular at $r=0$ there is a coordinate system in which it becomes {\quasireg}. Because the metric becomes {\quasireg} at $r=0$, the expanded Einstein equations are valid at $r=0$ too. But also Einstein's equation can be extended at $r=0$, because in this special case it becomes $G_{ab}=0$, the  {\schw} solution being a vacuum solution.
Hence, in this case we can just use the standard Einstein equations, of course in coordinates compatible with the coordinates \eqref{eq_coordinate_semireg}. Corollary \ref{thm_schw_quasireg} shows that the {\schw} singularity is {\quasireg} in any such coordinates. Since $S_{abcd}=E_{abcd}=0$, the only non-vanishing part of $R_{abcd}$ is the Weyl curvature $C_{abcd}=R_{abcd}$, which is smooth because $R_{abcd}$ is smooth.
\end{remark}

\begin{remark}
In the limit $m=0$, the {\schw} solution \eqref{eq_schw_schw} coincides with the Minkowski metric, which is regular at $r=0$. The event horizon singularity $r=2m$ merges with the $r=0$ singularity, and cancel one another. Because the {\schw} radius becomes $0$, the false singularity $r=0$ is not spacelike as in the case $m> 0$, but timelike.
In the case $m=0$, because there is no singularity at $r=0$, our coordinates \eqref{eq_coordinate_semireg}, rather than removing a (non-existent) singularity, introduce one. The new coordinates provide a double covering for the Minkowski spacetime, because $\tau$ extends beyond $r=0$ to negative values, in a way similar to the case described in \cite{Sto11f}.
\end{remark}

\begin{openproblem}
What can be said about the other stationary black hole solutions?
In \cite{Sto11f} and \cite{Sto11g} we showed that there are coordinate transformations which make the {\rn} metric and the {\kn} metric analytic at the singularity. This is already a big step, because it allows us to foliate with Cauchy hypersurfaces these spacetimes. Is it possible to find coordinate transformations which make them {\quasireg} too?
\end{openproblem}

\section{Conclusions}

An important problem in General Relativity is that of singularities. At singularities some of the quantities involved in the Einstein equation become infinite. But there are other quantities which are also invariant and in addition remain finite at a large class of singularities. In this paper it has been seen that translating the Einstein equation in terms of such quantities allows it to be extended at such singularities.

The Riemann tensor is, from geometric and linear-algebraic viewpoints, more fundamental than the Ricci tensor $R_{ab}$, which is just its trace. This suggests that the scalar part $S_{abcd}$ \eqref{eq_ricci_part_S} and the Ricci part $E_{abcd}$ \eqref{eq_ricci_part_E} of the Riemann curvature may be more fundamental than the Ricci tensor. Consequently, this justifies the study of an equation equivalent to Einstein's, but in terms of $E_{abcd}$ and $S_{abcd}$, instead of $R_{ab}$ and $R$. This is the expanded Einstein equation \eqref{eq_einstein_expanded}. The idea that $E_{abcd}$ is more fundamental than $R_{ab}$ seems to be suggested also by the electrovac solution, with the expanded Einstein equation \eqref{eq_stress_energy_maxwell_expanded}, and from which the electrovac Einstein equation is obtained by contraction.

To go from Einstein's equation to its expanded version we use the Kulkarni-Nomizu product \eqref{eq_kulkarni_nomizu}. To go back, we use contraction \eqref{eq_expanded_to_standard}. When the metric is {\nondeg}, these operations establish an equivalence between the standard and the expanded Einstein equations.

The question of whether the Ricci part of the Riemann tensor is more fundamental than the Ricci tensor may be irrelevant, or the answer may be debatable. But an important feature is that $E_{abcd}$ and $S_{abcd}$ can be defined in more general situations than $R_{ab}$ and $R$. Hence, the expanded Einstein equation is more general than the Einstein equation -- it makes sense even when the metric is degenerate, at least for a class of singularities named {\quasireg}. 

A brief investigation revealed that the class of {\quasireg} singularities is rich enough to contain some known singularities, which were already considered by researchers, but now can be understood in a unified framework. Among these there are the isotropic singularities, which are obtained by multiplying a regular metric with a scaling factor which is allowed to vanish. Another class is given by the {\flrw} singularities \cite{Sto12a}, and other warped product singularities. Even the {\schw} singularity (in proper coordinates which make the metric analytic \cite{Sto11e}) turns out to be \quasireg.

The fact that these apparently unrelated types of singularities turn out to be {\quasireg} suggests the following open question:
\begin{openproblem}
Are {\quasireg} singularities general enough to cover all possible singularities of General Relativity?
\end{openproblem}

\subsection*{Acknowledgments}

The author thanks the anonymous referees for the valuable comments and suggestions to improve the clarity and the quality of this paper.

%\bibliographystyle{plain}%{unsrt}%{amsalpha}%{amsplain}
%\bibliography{../bib/sing-gr_bib}

\begin{thebibliography}{10}

\bibitem{AT99i}
K.~Anguige and K.~P. Tod.
\newblock Isotropic cosmological singularities: {I}. {P}olytropic perfect fluid
  spacetimes.
\newblock {\em Ann. of Phys.}, 276(2):257--293, 1999.

\bibitem{AT99ii}
K.~Anguige and K.~P. Tod.
\newblock Isotropic cosmological singularities: {II}. {T}he {E}instein-{V}lasov
  system.
\newblock {\em Ann. of Phys.}, 276(2):294--320, 1999.

\bibitem{BESS87}
Arthur~L. Besse.
\newblock {\em {E}instein Manifolds, {E}rgebnisse der {M}athematik und ihrer
  {G}renzgebiete (3) [{R}esults in Mathematics and Related Areas (3)], vol.
  10}.
\newblock Berlin, New York: Springer-Verlag, 1987.

\bibitem{CN98}
C.~M. Claudel and K.~P. Newman.
\newblock {The Cauchy problem for quasi--linear hyperbolic evolution problems
  with a singularity in the time}.
\newblock {\em P. Roy. Soc. A-Math. Phy.}, 454(1972):1073, 1998.

\bibitem{corda2010removingBHsingularities}
C.~Corda and H.~J.~M. Cuesta.
\newblock {Removing Black Hole Singularities with Nonlinear Electrodynamics}.
\newblock {\em Mod. Phys. Lett. A}, 25(28):2423--2429, 2010.
\newblock \href{http://arxiv.org/abs/0905.3298v8}{arXiv:gr-qc/0905.3298}.

\bibitem{eddington1924comparison}
A.~S. Eddington.
\newblock {A Comparison of Whitehead's and Einstein's Formulae}.
\newblock {\em Nature}, 113:192, 1924.

\bibitem{finkelstein1958past}
D.~Finkelstein.
\newblock Past-future asymmetry of the gravitational field of a point particle.
\newblock {\em Phys. Rev.}, 110(4):965, 1958.

\bibitem{GHLF04}
S.~Gallot, D.~Hullin, and J.~Lafontaine.
\newblock {\em {R}iemannian Geometry}.
\newblock Springer-Verlag, Berlin, New York, 3rd edition, 2004.

\bibitem{Haw66i}
S.~W. Hawking.
\newblock The occurrence of singularities in cosmology.
\newblock {\em P. Roy. Soc. A-Math. Phy.}, 294(1439):511--521, 1966.

\bibitem{Haw66ii}
S.~W. Hawking.
\newblock The occurrence of singularities in cosmology. {II}.
\newblock {\em P. Roy. Soc. A-Math. Phy.}, 295(1443):490--493, 1966.

\bibitem{Haw67iii}
S.~W. Hawking.
\newblock The occurrence of singularities in cosmology. {III}. {C}ausality and
  singularities.
\newblock {\em P. Roy. Soc. A-Math. Phy.}, 300(1461):187--201, 1967.

\bibitem{HE95}
S.~W. Hawking and G.~F.~R. Ellis.
\newblock {\em {The Large Scale Structure of Space Time}}.
\newblock Cambridge University Press, 1995.

\bibitem{HP70}
S.~W. Hawking and R.~W. Penrose.
\newblock {The Singularities of Gravitational Collapse and Cosmology}.
\newblock {\em Proc. Roy. Soc. London Ser. A}, 314(1519):529--548, 1970.

\bibitem{ONe83}
B.~O'Neill.
\newblock {\em Semi-{R}iemannian Geometry with Applications to Relativity}.
\newblock Number 103 in Pure Appl. Math. Academic Press, New York-London, 1983.

\bibitem{Pen65}
R.~Penrose.
\newblock {Gravitational Collapse and Space-Time Singularities}.
\newblock {\em Phys. Rev. Lett.}, 14(3):57--59, 1965.

\bibitem{ST69}
I.~M. Singer and J.~A. Thorpe.
\newblock {The curvature of 4-dimensional Einstein spaces}.
\newblock In {\em {G}lobal Analysis ({P}apers in Honor of {K}. {K}odaira)},
  pages 355--365. {Princeton Univ. Press, Princeton, and Univ. Tokyo Press,
  Tokyo}, 1969.

\bibitem{Sto11a}
O.~C. Stoica.
\newblock On singular semi-{R}iemannian manifolds.
\newblock {\em To appear in Int. J. Geom. Methods Mod. Phys.}, May 2011.
\newblock \href{http://arxiv.org/abs/1105.0201}{arXiv:math.DG/1105.0201}.

\bibitem{Sto11b}
O.~C. Stoica.
\newblock Warped products of singular semi-{R}iemannian manifolds.
\newblock {\em Arxiv preprint math.DG/1105.3404}, May 2011.
\newblock \href{http://arxiv.org/abs/1105.3404}{arXiv:math.DG/1105.3404}.

\bibitem{Sto11f}
O.~C. Stoica.
\newblock Analytic {R}eissner-{N}ordstr{\"o}m singularity.
\newblock {\em \href{http://stacks.iop.org/1402-4896/85/i=5/a=055004}{Phys.
  Scr.}}, 85(5):055004, 2012.
\newblock \href{http://arxiv.org/abs/1111.4332}{arXiv:gr-qc/1111.4332}.

\bibitem{Sto12a}
O.~C. Stoica.
\newblock Beyond the {F}riedmann-{L}ema{\^i}tre-{R}obertson-{W}alker {B}ig
  {B}ang singularity.
\newblock {\em Commun. Theor. Phys.}, 58(4):613--616, March 2012.
\newblock \href{http://arxiv.org/abs/1203.1819}{arXiv:gr-qc/1203.1819}.

\bibitem{Sto12e}
O.~C. Stoica.
\newblock
  \href{http://www.degruyter.com/view/j/auom.2012.20.issue-2/v10309-012-0050-3/v10309-012-0050-3.xml}{Spacetimes
  with Singularities}.
\newblock {\em An. {\c S}t. Univ. Ovidius Constan{\c t}a}, 20(2):213--238, July
  2012.
\newblock \href{http://arxiv.org/abs/1108.5099}{arXiv:gr-qc/1108.5099}.

\bibitem{Sto11e}
O.~C. Stoica.
\newblock Schwarzschild singularity is semi-regularizable.
\newblock {\em \href{http://dx.doi.org/10.1140/epjp/i2012-12083-1}{Eur. Phys.
  J. Plus}}, 127(83):1--8, 2012.
\newblock \href{http://arxiv.org/abs/1111.4837}{arXiv:gr-qc/1111.4837}.

\bibitem{Sto11h}
O.~C. Stoica.
\newblock {B}ig {B}ang singularity in the
  {F}riedmann-{L}ema{\^i}tre-{R}obertson-{W}alker spacetime.
\newblock {\em The International Conference of Differential Geometry and
  Dynamical Systems}, October 2013.
\newblock \href{http://arxiv.org/abs/1112.4508}{arXiv:gr-qc/1112.4508}.

\bibitem{Sto11g}
O.~C. Stoica.
\newblock {K}err-{N}ewman solutions with analytic singularity and no closed
  timelike curves.
\newblock {\em To appear in U.P.B. Sci. Bull., Series A}, 2013.
\newblock \href{http://arxiv.org/abs/1111.7082}{arXiv:gr-qc/1111.7082}.

\bibitem{Sto12c}
O.~C. Stoica.
\newblock On the {W}eyl curvature hypothesis.
\newblock {\em Ann. of Phys.}, 338:186--194, November 2013.
\newblock \href{http://arxiv.org/abs/1203.3382}{arXiv:gr-qc/1203.3382}.

\bibitem{Tod87}
K.~P. Tod.
\newblock {Quasi-local Mass and Cosmological Singularities}.
\newblock {\em Class. Quant. Grav.}, 4:1457, 1987.

\bibitem{Tod90}
K.~P. Tod.
\newblock {Isotropic Singularities and the $\gamma=2$ Equation of State}.
\newblock {\em Class. Quant. Grav.}, 7:L13--L16, 1990.

\bibitem{Tod91}
K.~P. Tod.
\newblock {Isotropic Singularities and the Polytropic Equation of State}.
\newblock {\em Class. Quant. Grav.}, 8:L77, 1991.

\bibitem{Tod92}
K.~P. Tod.
\newblock {Isotropic Singularities}.
\newblock {\em Rend. Sem. Mat. Univ. Politec. Torino}, 50:69--93, 1992.

\bibitem{Tod02}
K.~P. Tod.
\newblock {Isotropic Cosmological Singularities}.
\newblock {\em The Conformal Structure of Space-Time}, pages 123--134, 2002.

\bibitem{Tod03}
K.~P. Tod.
\newblock {Isotropic Cosmological Singularities: Other Matter Models}.
\newblock {\em Class. Quant. Grav.}, 20:521, 2003.

\end{thebibliography}

\end{document}